%% file: main.tex
\documentclass[a4paper]{article}

\usepackage[english]{babel}
\usepackage[utf8x]{inputenc}
\usepackage[T1]{fontenc}

\usepackage[a4paper,top=3cm,bottom=2cm,left=3cm,right=3cm,marginparwidth=1.75cm]{geometry}

\usepackage{amsmath}
\usepackage{amssymb}
\usepackage{graphicx}
\usepackage[colorinlistoftodos]{todonotes}
\usepackage[colorlinks=true, allcolors=blue]{hyperref}
\usepackage{amsthm}
\newtheorem{theorem}{Theorem}

\newtheorem{conjecture}{Conjecture}
\newtheorem{proposition}{Proposition}
\newtheorem{claim}{Claim}
\newtheorem{lemma}{Lemma}
\newtheorem{definition}{Definition}

\newtheorem{corollary}{Corollary}
\newtheorem{fact}{Fact}
\newtheorem{observation}{Observation}
\newtheorem{remark}{Remark}
\newtheorem{algorithm}{Algorithm}
\usepackage{algpseudocode}
\usepackage{algorithm}
\usepackage{comment}
\newcommand{\gnote}[1]{}
\newcommand{\ynote}[1]{}

\title{Being Corrupt Requires Being Clever, But Detecting Corruption Doesn't}
\author{Yan Jin \footnote{Email: \url{yjin1@mit.edu}. Partially supported by Institute for Data, Systems and Society Fellowship and ARO MURI award No.W911NF-12-1-0509.}\\MIT  
\and Elchanan Mossel \footnote{Email: \url{elmos@mit.edu}. Partially supported by awards ONR N00014-16-1-2227, NSF CCF1665252
and DMS-1737944.}\\MIT 
\and Govind Ramnarayan\footnote{Email: \url{govind@mit.edu}. Partially supported by awards
    NSF CCF 1665252 and DMS-1737944.}\\MIT}

\begin{document}
\maketitle

\input{abstract}

\input{intro}

\input{preliminaries}
\input{results}

\input{hardness}

\input{directed}

\section{Acknowledgements}
We would like to thank Paxton Turner and Vishesh Jain for useful conversations about corruption detection over the past year. We would like to thank Pasin Manurangsi for pointing us to the Austrin-Pitassi-Wu inapproximability result for treewidth.

\bibliographystyle{alpha}
\bibliography{sample}

\appendix
\input{appendix}

\end{document}

%% file: abstract.tex
\begin{abstract}
We consider a variation of the problem of corruption detection on networks posed by Alon, Mossel, and Pemantle '15. In this model, each vertex of a graph can be either truthful or corrupt. Each vertex reports about the types (truthful or corrupt) of all its neighbors to a central agency, where truthful nodes report the true types they see and corrupt nodes report adversarially. The central agency aggregates these reports and attempts to find a single truthful node. Inspired by real auditing networks, we pose our problem for arbitrary graphs and consider corruption through a computational lens. We identify a key combinatorial parameter of the graph $m(G)$, which is the minimal number of corrupted agents needed to prevent the central agency from identifying a single truthful node. 
We give an efficient (in fact, linear time) algorithm for the central agency to identify a truthful node that is successful whenever the number of corrupt nodes is less than $m(G)/2$. On the other hand, we prove that for any constant $\alpha > 1$,  it is NP-hard to find a subset of nodes $S$ in $G$ such that corrupting $S$ prevents the central agency from finding one truthful node and $|S| \leq \alpha m(G)$, assuming the Small Set Expansion Hypothesis (Raghavendra and Steurer, STOC '10). We conclude that being corrupt requires being clever, while detecting corruption does not.

Our main technical insight is a relation between the minimum number of corrupt nodes required to hide all truthful nodes and a certain notion of vertex separability for the underlying graph. Additionally, this insight lets us design an efficient algorithm for a corrupt party to decide which graphs require the fewest corrupted nodes, up to a multiplicative factor of $O(\log n)$.
\end{abstract}

%% file: intro.tex
\section{Introduction}
\label{sec:intro}
\subsection{Corruption Detection and Problem Set-up}
\label{sec:problem-set-up}

We study the problem of identifying truthful nodes in networks, in the model of \emph{corruption detection on networks} posed by Alon, Mossel, and Pemantle~\cite{AMP15}. In this model, we have a network represented by a (possibly directed) graph. Nodes can be \emph{truthful} or \emph{corrupt}. Each node audits its outgoing neighbors to see whether they are truthful or corrupt, and sends reports of their identities to a central agency. The central agent, who is not part of the graph, aggregates the reports and uses them to identify truthful and corrupt nodes. Truthful nodes report truthfully (and correctly) on their neighbors, while corrupt nodes have no such restriction: they can assign arbitrary reports to their neighbors, regardless of whether their neighbors are truthful or corrupt, and coordinate their efforts with each other to prevent the central agency from gathering useful information.

In~\cite{AMP15}, the authors consider the problem of recovering the identities of almost all nodes in a network in the presence of many corrupt nodes; specifically, when the fraction of corrupt nodes can be very close to $1/2$. They call this the \emph{corruption detection} problem. They show that the central agency can recover the identity of most nodes correctly even in certain bounded-degree graphs, as long as the underlying graph is a sufficiently good expander. The required expansion properties are known to hold for a random graph or Ramanujan graph of sufficiently large (but constant) degree, which yields undirected graphs that are amenable to corruption detection. Furthermore, they show that some level of expansion is necessary for identifying truthful nodes, by demonstrating that the corrupt nodes can stop the central agency from identifying any truthful node when the graph is a very bad expander (e.g. a cycle), even if the corrupt nodes only make up $0.01$ fraction of the network. 

This establishes that very good expanders are very good for corruption detection, and very bad expanders can be very bad for corruption detection. We note that this begs the question of how effective graphs that do not fall in either of these categories are for corruption detection. In the setting of~\cite{AMP15}, we could ask the following: given an \emph{arbitrary} undirected graph, what is the smallest number of corrupt nodes that can prevent the identification of almost all nodes? When there are fewer than this number, can the central agency \emph{efficiently} identify almost all nodes correctly? Alon, Mossel, and Pemantle study these questions for the special cases of highly expanding graphs and poorly expanding graphs, but do not address general graphs.

Additionally,~\cite{AMP15} considers corruption detection when the corrupt agencies can choose their locations and collude arbitrarily, with no bound on their computational complexity. This is perhaps overly pessimistic: after all, it is highly unlikely that corrupt agencies can solve NP-hard problems efficiently
and if they can, thwarting their covert operations is unlikely to stop their world domination. We suggest a model that takes into account computational considerations, by factoring in the computation time required to select the nodes in a graph that a corrupt party chooses to control. This yields the following question from the viewpoint of a corrupt party: given a graph, can a corrupt party compute the smallest set of nodes it needs to corrupt \emph{in polynomial time}? 

In addition to being natural from a mathematical standpoint, these questions are also well-motivated socially. It would be na\"{i}ve to assert that we can weed out corruption in the real world by simply designing auditing networks that are expanders. Rather, these networks may already be formed, and infeasible to change in a drastic way. Given this, we are less concerned with finding certain graphs that are good for corruption detection, but rather discerning how good \emph{existing} graphs are; specifically, how many corrupt nodes they can tolerate. In particular, since the network structure could be out of the control of the central agency, algorithms for the central agency to detect corruption on arbitrary graphs seem particularly important.

It is also useful for the \emph{corrupt} agency to have an algorithm with guarantees for any graph. Consider the following example of a corruption detection problem from the viewpoint of a corrupt organization. Country A wants to influence policy in country B, and wants to figure out the most efficient way to place corrupted nodes within country B to make this happen. However, if the central government of B can confidently identify truthful nodes, they can weight those nodes' opinions more highly, and thwart country A's plans. Hence, the question country A wants to solve is the following: given the graph of country B, can country A compute the optimal placement of corrupt nodes to prevent country B from finding truthful nodes? We note that in this question, too, the graph of country B is fixed, and hence, country A would like to have an algorithm that takes as input \emph{any} graph and computes the optimal way to place corrupt nodes in order to hide all the truthful nodes. 

We study the questions above for a variant of the corruption detection problem in~\cite{AMP15}, in which the goal of the central agency is to find a single truthful node. While this goal is less ambitious than the goal of identifying almost all the nodes, we think it is a very natural question in the context of corruption. For one, if the central agency can find a single truthful node, they can use the trusted reports from that node to identify more truthful and corrupt nodes that it might be connected to. The central agency may additionally weight the opinions of the truthful nodes more when making policy decisions (as alluded to in the example above), and can also incentivize truthfulness by rewarding truthful nodes that it finds and giving them more influence in future networks if possible (by increasing their out-degrees). Moreover, our proofs and results extend to finding larger number of truthful nodes as we discuss below. 

Our results stem from a tie between the problem of finding a single truthful node in a graph and a measure of vertex separability of the graph. This tie not only yields an efficient and relatively effective algorithm for the central agency to find a truthful node, but also allows us to relate corrupt party's strategy to the problem of finding a good vertex separator for the graph. Hence, by analyzing the purely graph-theoretic problem of finding a good vertex separator, we can characterize the difficulty of finding a good set of nodes to corrupt. Similar notions of vertex separability have been studied previously (e.g.~\cite{Lee17, ORS07, BMN15}), and we prove NP-hardness for the notion relevant to us assuming the Small Set Expansion Hypothesis (SSEH). The \emph{Small Set Expansion Hypothesis} is a hypothesis posed by Raghavendra and Steurer~\cite{RS10} that is closely related to the famous Unique Games Conjecture of Khot~\cite{Khot02}. In fact,~\cite{RS10} shows that the SSEH implies the Unique Games Conjecture. The SSEH yields hardness results that are not known to follow directly from the UGC, especially for graph problems like sparsest cut and treewidth (\cite{RST12} and \cite{APW12} respectively), among others.

\subsection{Our Results} 
\label{sec:results}
We now outline our results more formally. We analyze the variant of corruption detection where the central agency's goal is to find a single truthful node. First, we study how effectively the central agency can identify a truthful node on an arbitrary graph, given a set of reports. Given an undirected graph\footnote{Unless explicitly specified, all graphs are undirected by default.} $G$, we let $m(G)$ denote the minimal number of corrupted nodes required to stop the central agency from finding a truthful node, where the minimum is taken over all strategies of the corrupt party (not just computationally bounded ones). We informally call $m(G)$ the ``critical'' number of corrupt nodes for a graph $G$. Then, we show the following:

\newtheorem*{thm:1}{Theorem \ref{thm:1}}
\begin{thm:1}
Fix a graph $G$ and suppose that the corrupt party has a budget $b \leq m(G) / 2$. Then the central agency can identify a truthful node, regardless of the strategy of the corrupt party, and without knowledge of either $m(G)$ or $b$. Furthermore, the central agency's algorithm runs in linear time (in the number of edges in the graph $G$).
\end{thm:1}


Next, we consider the question from the viewpoint of the corrupt party: can the corrupt party efficiently compute the most economical way to allocate nodes to prevent the central agency from finding a truthful node? Concretely, we focus on a natural decision version of the question: given a graph $G$ and a upper bound on the number of possible corrupted nodes $k$, can the corrupt party prevent the central agency from finding a truthful node? 

We actually focus on an easier question: can the corrupt party accurately compute $m(G)$, the minimum number of nodes that they need to control to prevent the central agency from finding a truthful node? Not only do we give evidence that computing $m(G)$ exactly is computationally hard, but we also provide evidence that $m(G)$ is hard to approximate. Specifically, we show that approximating $m(G)$ to any constant factor is NP-hard under the Small Set Expansion Hypothesis (SSEH); or in other words, that it is SSE-hard.

\newtheorem*{thm:2}{Theorem \ref{thm:2}}
\begin{thm:2}
For every $\beta > 1$, there is a constant $\epsilon > 0$ such that the following is true. Given a graph $G = (V,E)$, it is SSE-hard to distinguish between the case where $m(G) \leq \epsilon \cdot |V|$ and $m(G) \geq \beta \cdot \epsilon \cdot |V|$. Or in other words, the problem of approximating the critical number of corrupt nodes for a graph to within any constant factor is SSE-hard.
\end{thm:2}

This Theorem immediately implies the following Corollary \ref{corr:1}.

\newtheorem*{corr:1}{Corollary \ref{corr:1}}
\begin{corr:1}
Assume the SSE Hypothesis and that P $\neq$ NP. Fix any $\beta>1$. There does not exist a polynomial-time algorithm that takes as input an arbitrary graph $G = (V,E)$ and outputs a set of nodes $S$ with size $|S|\leq O(\beta \cdot m(G))$, such that corrupting $S$ prevents the central agency from finding a truthful node.
\end{corr:1}


We note that in Corollary \ref{corr:1}, the bad party's input is only the graph $G$: 
specifically, they do not have knowledge about the value of $m(G)$.

Our proof for Theorem~\ref{thm:2} is similar to the proof of Austrin, Pitassi, and Wu~\cite{APW12} for the SSE-hardness of approximating treewidth. This is not a coincidence: in fact, the ``soundness'' in their reduction involves proving that their graph does not have a good $1/2$ vertex separator, where the notion of vertex separability (from~\cite{BGHK95}) is very related to the version we use to categorize the problem of hiding a truthful vertex. We give the proof of Theorem~\ref{thm:2} in Section~\ref{sec:hardness-proof}.

However, if one allows for an approximation factor of $O(\log |V|)$, then $m(G)$ can be approximated efficiently. Furthermore, this yields an approximation algorithm that the corrupt party can use to find a placement that hinders detection of a truthful node.
\newtheorem*{thm:3}{Theorem \ref{thm:3}}
\begin{thm:3}
There is a polynomial-time algorithm that takes as input a graph $G = (V,E)$ and outputs a set of nodes $S$ with size $|S|\leq O(\log |V|\cdot m(G))$, such that corrupting $S$ prevents the central agency from finding a truthful node.
\end{thm:3}


The proof of Theorem~\ref{thm:3}, given in Section~\ref{sec:hardness-proof}, uses a bi-criterion approximation algorithm for the $k$-vertex separator problem given by~\cite{Lee17}. As alluded to in Section~\ref{sec:problem-set-up}, Theorems~\ref{thm:2} and~\ref{thm:3} both rely on an approximate characterization of $m(G)$ in terms of a measure of vertex separability of the graph $G$, which we give in Section~\ref{sec:main-results}. 

Additionally, we note that we can adapt Theorems~\ref{thm:1} and~\ref{thm:2} (as well as Corollary~\ref{corr:1}) to a more general setting, where the central agency wants to recover some arbitrary number of truthful nodes, where the number of nodes can be proportional to the size of the graph. We describe how to modify our proofs to match this more general setting in Section~\ref{sec:find-many-truthful}.

Together, Theorems~\ref{thm:1} and~\ref{thm:2} uncover a surprisingly positive result for corruption detection: it is computationally easy for the central agency to find a truthful node when the number of corrupted nodes is only somewhat smaller than the ``critical'' number for the underlying graph, but it is in general computationally hard for the corrupt party to hide the truthful nodes even when they have a budget that far exceeds the ``critical'' number for the graph.

\paragraph{Results for Directed Graphs}
As noted in~\cite{AMP15}, it is unlikely that real-world auditing networks are undirected. For example, it is likely that the FBI has the authority to audit the Cambridge police department, but it is also likely that the reverse is untrue. Therefore, we would like the central agency to be able to find truthful nodes in directed graphs in addition to undirected graphs. We notice that the algorithm we give in Theorem~\ref{thm:1} extends naturally to directed graphs. 

\newtheorem*{thm:4}{Theorem \ref{thm:4}}
\begin{thm:4}
Fix a directed graph $D$ and suppose that the corrupt party has a budget $b \leq m(D) / 2$. Then the central agency can identify a truthful node, regardless of the strategy of the corrupt party, and without the knowledge of either $m(D)$ or $b$. Furthermore, the central agency's algorithm runs in linear time.
\end{thm:4}


The proof of Theorem~\ref{thm:4} is similar to the proof of Theorem~\ref{thm:1}, and effectively relates the problem of finding a truthful node on directed graphs to a similar notion of vertex separability, suitably generalized to directed graphs.

\paragraph{Results for Finding An Arbitrary Number of Good Nodes}
In fact, the problem of finding one good node is just a special case of finding an arbitrary number of good nodes, $g$, on the graph $G$. We define $m(G,g)$ as the minimal number of bad nodes required to prevent the identification of $g$ good nodes on the graph $G$. We relate it to an analogous vertex separation notion, and prove the following two theorems, which are extensions of Theorems~\ref{thm:1} and~\ref{thm:2} to this setting.

\newtheorem*{thm:efficient_alg}{Theorem \ref{thm:efficient_alg}}
\begin{thm:efficient_alg}
Fix a graph $G$ and the number of good nodes to recover, $g$. Suppose that the corrupt party has a budget $b\leq m(G,g) / 2$. If $g<|V|-2b,$ then the central agency can identify $g$ truthful nodes, regardless of the strategy of the corrupt party, and without knowledge either of $m(G,g)$ or $b$. Furthermore, the central agency's algorithm runs in linear time. 
\end{thm:efficient_alg}

\newtheorem*{thm:sse_hard}{Theorem \ref{thm:sse_hard}}
\begin{thm:sse_hard}
For every $\beta>1$ and every $0<\delta<1$ 
, there is a constant $\epsilon>0$ such that the following is true. Given a graph $G=(V,E)$, it is SSE-hard to distinguish between the case where $m(G,\delta |V|)\leq \epsilon \cdot |V|$ and $m(G,\delta|V|)\geq \beta \cdot \epsilon \cdot |V|.$ Or in other words, the problem of approximating the critical number of corrupt nodes such that it is impossible to find $\delta |V|$ good nodes within any constant factor is SSE-hard.
\end{thm:sse_hard}

The proof of Theorem~\ref{thm:sse_hard} is similar to the proof of Theorem~\ref{thm:1}, and the hardness of approximation proof also relies on the same graph reduction and SSE conjecture. Proofs are presented in Section \ref{sec:find-many-truthful}.

\subsection{Related Work}
\label{sec:related-work}
The model of corruptions posed by~\cite{AMP15} is identical to a model first suggested by Perparata, Metze, and Chien~\cite{PMC67}, who introduced the model in the context of detecting failed components in digital systems. This work (as well as many follow-ups, e.g.~\cite{Kameda75,KR80}) looked at the problem of characterizing which networks can detect a certain number of corrupted nodes. Xu and Huang~\cite{XH95} give necessary and sufficient conditions for identifying a single corrupted node in a graph, although their characterization is not algorithmically efficient. There are many other works on variants of this problem (e.g. \cite{Sull84,DM84}), including recovering node identities with one-sided or two-sided error probabilities in the local reports~\cite{MH76} and adaptively finding truthful nodes~\cite{HA74}.

We note that our model of a computationally bounded corrupt party and our stipulation that the graph is fixed ahead of time rather than designed by the central agency, which are our main contributions to the model, seem more naturally motivated in the setting of corruptions than in the setting of designing digital systems. Even the question of identifying a single truthful node could be viewed as more naturally motivated in the setting of corruptions than in the setting of diagnosing systems. We believe there are likely more interesting theoretical questions to be discovered by approaching the PMC model through a corruptions lens. 

The identifiability of a single node in the corruptions setting was studied in a recent paper of Mukwembi and Mukwembi~\cite{MM17}. They give a linear time greedy algorithm to recover the identify of a single node in many graphs, \emph{provided that corrupt nodes always report other corrupt nodes as truthful}. Furthermore, this assumption allows them to reduce identifying all nodes to identifying a single node. They argue that such an assumption is natural in the context of corruptions, where corrupt nodes are selfishly incentivized not to out each other. However, in our setting, corrupt nodes can not only betray each other, but are in fact incentivized to do so for the good of the overarching goal of the corrupt party (to prevent the central agency from identifying a truthful node). Given~\cite{MM17}, it is not a surprise that the near-optimal strategies we describe for the corrupt party in this paper crucially rely on the fact that the nodes can report each other as corrupt.

Our problem of choosing the best subset of nodes to corrupt bears intriguing similarities to the problem of influence maximization studied by~\cite{KKT15}, where the goal is to find an optimal set of nodes to target in order to maximize the adoption of a certain technology or product. It is an interesting question to see if there are further similarities between these two areas. Additionally, social scientists have studied corruption extensively (e.g.\cite{Fjeldstad03}, \cite{Nielsen03}), though to the best of our knowledge they have not studied it in the graph-theoretic way that we do in this paper.

\subsection{Comparison to Corruption in Practice}
Finally, we must address the elephant in the room. Despite our theoretical results, corruption \emph{is} prevalent in many real-world networks, and yet in many scenarios it is not easy to pinpoint even a single truthful node. One reason for that is that some of assumptions do not seem to hold in some real world networks. For example, we assume that audits from the truthful nodes are not only non-malicious, but also perfectly reliable. In practice this assumption is unlikely to be true: many truthful nodes could be non-malicious but simply unable to audit their neighbors accurately. Further assumptions that may not hold in some scenarios include the notion of a central agency that is both uncorrupted and has access to reports from every agency, and possibly even the assumption that the number of corrupt nodes is less than $|V| / 2$. In addition, networks $G$ may have very low critical numbers $m(G)$ in practice. 
For example, there could be a triangle (named, ``President'', ``Congress'' and ``Houses'') that is all corrupt and cannot be audited by any agent outside the triangle.  
 It is thus plausible that a corrupt party could use the structure of realistic auditing networks for their corruption strategy to overcome our worst-case hardness result. 

While this points to some shortcomings of our model, it also points out ways to change policy that would potentially bring the real world closer to our idealistic scenario, where a corrupt party has a much more difficult computational task than the central agency. For example, we can speculate that perhaps information should be gathered by a transparent  centralized agency, that significant resources should go into ensuring that the centralized agency is not corrupt, and that networks ought to have good auditing structure (without important agencies that can be audited by very few nodes). 

%% file: preliminaries.tex
\section{Preliminaries}
\subsection{General Preliminaries}
We denote undirected graphs by $G=(V, E)$, where $V$ is the vertex set of the graph and $E$ is the edge set. We denote directed graphs by $D = (V,E_D)$. When the underlying graph is clear, we may drop the subscripts. Given a vertex $u$ in an undirected graph $G$, we let $\mathcal{N}(u)$ denote the \emph{neighborhood} (set of neighbors) of the vertex in $G$. Similarly, given a vertex $u$ in a directed graph $D$, let $\mathcal{N}(u)$ denote the set of \emph{outgoing} neighbors of $u$: that is, vertices $v \in V$ such that $(u,v) \in E_D$.

\subsubsection{Vertex Separator}
\begin{definition}
\textbf{(k-vertex separator)}(\cite{ORS07},\cite{BMN15}) For any $k \geq 0$, we say a subset of vertices $U\subseteq V$ is k-vertex separator of a graph $G$, if after removing $U$ and incident edges, the remaining graph forms a union of connected components, each of size at most $k$.

Furthermore, let $$S_G(k)=\min \big( |U| : U \textrm{ is a } k\textrm{-vertex separator of } G \big)$$ denote the size of the minimal $k$-vertex separator of graph $G$.
\end{definition}

\subsubsection{Small Set Expansion Hypothesis}
In this section we define the Small Set Expansion (SSE) Hypothesis introduced in \cite{RS10}. Let $G=(V,E)$ be an undirected $d$-regular graph. 
\begin{definition}[Normalized edge expansion]
For a set $S\subseteq V$ of vertices, denote $\Phi_G(S)$ as the normalized edge expansion of $S$,
\begin{equation*}
\Phi_G(S)=\frac{|E(S,V\backslash S)|}{d|S|},
\end{equation*}
where $|E(S,V\backslash S)|$ is the number of edges between $S$ and $V\backslash S.$
\end{definition}

\textit{The Small Set Expansion Problem} with parameters $\eta$ and $\delta$, denoted SSE($\eta,\delta$), asks whether $G$ has a small set $S$ which does not expand or all small sets of $G$ are highly expanding. 
\begin{definition}[(SSE($\eta,\delta$))]
Given a regular graph $G=(V,E),$ distinguish between the following two cases:
\begin{itemize}
\item \textbf{Yes} There is a set of vertices $S\subseteq V$ with $S=\delta |V|$ and $\Phi_G(S)\leq \eta$

\item \textbf{No} For every set of vertices $S\subseteq V$ with $S=\delta |V|$ it holds that $\Phi_G(S)\geq 1-\eta$
\end{itemize}
\end{definition}

The \textit{Small Set Expansion Hypothesis} is the conjecture that deciding SSE($\eta,\delta$) is NP-hard.

\begin{conjecture}[Small Set Expansion Hypothesis \cite{RS10}]
For every $\eta>0$, there is a $\delta>0$ such that SSE($\eta,\delta$) is NP-hard.
\label{sse}
\end{conjecture}

We say that a problem is \textit{SSE-hard} if it is at least as hard to solve as the SSE problem. The form of conjecture most relevant to our proof is the following ``stronger'' form of the SSE Hypothesis. \cite{RST12} showed that the SSE-problem can be reduced to a quantitatively stronger form of itself. In order to state this version, we first need to define the \emph{Gaussian noise stability}.

\begin{definition}(Gaussian Noise Stability)
\label{def:gauss-noise}
Let $\rho\in[-1,1]$. Define $\Gamma_\rho:[0,1]\mapsto [0,1]$ by
\begin{equation*}
\Gamma_\rho(\mu) = Pr[X\leq \Phi^{-1}(\mu) \wedge Y\leq \Phi^{-1}(\mu)]
\end{equation*}
where $X$ and $Y$ are jointly normal random variables with mean $0$ and covariance matrix $\begin{pmatrix}1 & \rho \\ \rho & 1\end{pmatrix}.$
\end{definition}
The only fact that we will use for stating the stronger form of SSEH is the asymptotic behavior of $\Gamma_\rho(\mu)$ when $\rho$ is close to $1$ and $\mu$ bounded away from 0.

\begin{fact}
\label{fact:gauss-noise}
There is a constant $c>0$ such that for all sufficiently small $\epsilon$ and all $\mu\in [1/10,1/2],$\footnote{Note that the lower bound on $\mu$ can be taken arbitrarily close to $0$. So the statement holds with $\mu \in [\epsilon',1/2]$ for any constant $\epsilon'>0$.}
\begin{equation*}
\Gamma_{1-\epsilon}(\mu) \leq \mu(1-c\sqrt{\epsilon}).
\end{equation*}
\end{fact}

\begin{conjecture}[SSE Hypothesis, Equivalent Formulation \cite{RST12}]
For every integer $q>0$ and $\epsilon,\gamma>0$, it is NP-hard to distinguish between the following two cases for a given regular graph $G=(V,E)$:
\begin{itemize}
\item \textbf{Yes} There is a partition of $V$ into $q$ equi-sized sets $S_1,\cdots,S_q$ such that $\Phi_G(S_i)\leq 2\epsilon$ for every $1\leq i\leq q.$
\item \textbf{No} For every $S\subseteq V,$ letting $\mu=|S|/|V|$, it holds that $\Phi_G(S)\geq 1-(\Gamma_{1-\epsilon/2}(\mu)+\gamma)/\mu,$ 
\end{itemize}
where the $\Gamma_{1-\epsilon/2}(\mu)$ is the Gaussian noise stability.
\label{SSE}
\end{conjecture}

We present two remarks about the Conjecture \ref{SSE} from \cite{APW12}, which are relevant to our proof of Theorem~\ref{thm:2}.

\begin{remark}\cite{APW12}
The \textbf{Yes} instance of Conjecture \ref{SSE} implies that the number of edges leaving each $S_i$ is at most $4\epsilon|E|/q,$ so the total number of edges not contained in one of the $S_i$ is at most $2\epsilon|E|.$ 
\label{rmk1}
\end{remark}

\begin{remark}\cite{APW12}
The \textbf{No} instance of Conjecture \ref{SSE} implies that for $\epsilon$ sufficiently small, there exists some constant $c'$ such that $\Phi_G(S)\geq c'\sqrt{\epsilon},$ provided that $\mu \in [1/10,1/2]$ and setting $\gamma \leq \sqrt{\epsilon}$. In particular, $|E(S,V\backslash S)|\geq \Omega(\sqrt{\epsilon}|E|),$ for any $|V|/10 \leq |S| \leq 9|V|/10.$  \footnote{Recall that Fact \ref{fact:gauss-noise} is true for $\mu \in [\epsilon',1/2]$ for any constant $\epsilon'>0$. Therefore, Remark \ref{rmk2} can be strengthened and states, for any $\epsilon'|V|\leq|S|\leq (1-\epsilon')|V|,$ $|E(S,V\backslash S)| \geq \Omega (\sqrt{\epsilon}|E|)$. This will be a useful fact for proving hardness of approximation of $m(G,g)$ for finding many truthful nodes in Section \ref{sec:find-many-truthful}.} 
\label{rmk2}
\end{remark}

Remark \ref{rmk1} follows from the definition of normalized edge expansion and the fact that sum of degree is two times number of edges. Remark \ref{rmk2} follows from Fact \ref{fact:gauss-noise}. The strong form of SSE Hypothesis \ref{SSE}, Remark \ref{rmk1}, and Remark \ref{rmk2} will be particularly helpful for proving our SSE-hardness of approximation result (Theorem~\ref{thm:2}).

\subsection{Preliminaries for Corruption Detection on Networks}

We model networks as directed or undirected graphs, where each vertex in the network can be one of two types: truthful or corrupted. At times, we will informally call truthful vertices ``good'' and corrupt vertices ``bad.'' We say that the corrupt party has \textit{budget} $b$ if it can afford to corrupt at most $b$ nodes of the graph. Given a vertex set $V$, and a budget $b$, the corrupt entity will choose to control a subset of nodes $B \subseteq V$ under the constraint $|B|\leq b$. The rest of the graph remains as truthful vertices, i.e., $T = V\backslash B \subseteq V$. We assume that there are more truthful than corrupt nodes ($b<|V|/2$). It is easy to see that in the case where $|B| \geq |T|$, the corrupt nodes can prevent the identification of even one truthful node, by simulating truthful nodes (see e.g.~\cite{AMP15}).

Each node audits and reports its (outgoing) neighbors’ identities. That is, each vertex $u \in V$ will report the type of each $v \in \mathcal{N}(u)$, which is a vector in $\{0,1\}^{|\mathcal{N}(u)|}$. Truthful nodes always report the truth, i.e., it reports its neighbor $v\in T$ if $v$ is truthful, $v\in B$ if $v$ is corrupt. The corrupt nodes report their neighbors' identities adversarially. In summary, a strategy of the bad agents is composed of a strategy to take over at most $b$ nodes on the graph, and reports on the nodes that neighbor them. 

\begin{definition}[\textbf{Strategy for a corrupt party}]
A strategy for the corrupt party is a function that maps a graph $G$ and budget $b$ to a subset of nodes $B$ with size $|B| \leq b$, and a set of reports that each node $v\in B$ gives about its neighboring nodes, $\mathcal{N}(v).$
\end{definition}

\begin{definition}[\textbf{Computationally bounded corrupt party}]
We say that the corrupt party is computationally bounded if its strategy can only be a polynomial-time computable function. 
\end{definition}

The task for the central agency is to find a good node on this corrupted network, based on the reports. It is clear that the more budget the corrupt party has, the harder the task of finding one truthful node becomes. It was observed in~\cite{AMP15} that, for any graph, it is not possible to find one good node if $b \geq |V|/2$. If $b=0$, it is clear that the entire set $V$ is truthful. Therefore, given an arbitrary graph $G$, there exists a critical number $m(G)$, such that if the bad party has budget lower than $m(G)$, it is always possible to find a good node; if the bad party has budget greater than or equal to $m(G),$ it may not be possible to find a good node. In light of this, we define the critical number of bad nodes on a graph $G$. First, we formally define what we mean when we say it is impossible to find a truthful node on a graph $G$.


\begin{definition}[\textbf{Impossibility of finding one truthful node}]
\label{def:impossible-good}
Given a graph $G=(V,E)$, the bad party's budget $b$ and reports, we say that it is \emph{impossible to identify one truthful node} if for every $v\in V$ there is a configuration of the identities of the nodes where $v$ is bad, and the configuration is consistent with the given reports, and consists of fewer than or equal to $b$ bad nodes.
\end{definition}

\begin{definition}[\textbf{Critical number of bad nodes on a graph $G$, $m(G)$}]
\label{def:critical-num}
Given an arbitrary graph $G=(V,E)$, we define \emph{$m(G)$} as the minimum number $b$ such that there is a way to distribute $b$ corrupt nodes and set their corresponding reports such that it is impossible to find \textit{one} truthful node on the graph $G$, given $G$, the reports and that the bad party's budget is at most $b$.
\end{definition}

For example, for a star graph $G$ with $|V|\geq 5$, the critical number of bad nodes is $m(G) =2$. If there is at most $1$ corrupt node on $G$, the central agency can always find a good node, thus $m(G)\neq 1$. If there are at most $2$ bad nodes on $G$, then the bad party can control the center node and one of the leaves. It is impossible for central agency to find one good node.

Given a graph $G$, by definition there exists some set of $m(G)$ nodes that can make it impossible to find a good node if they are corrupted. However, this does not mean that the corrupt party can necessarily find this set in polynomial time. Indeed, Theorem~\ref{thm:2} establishes that they cannot always find this set in polynomial time if we assume the SSE Hypothesis (Conjecture~\ref{SSE}) and that P $\neq$ NP.

%% file: results.tex
\section{Proofs of Theorems~\ref{thm:1}, \ref{thm:2}, and \ref{thm:3}}
\label{sec:main-results}
In the following section, we state our main results by first presenting the close relation of our problem to the $k$-vertex separator problem. Then we use this characterization to prove Theorem~\ref{thm:1}. This characterization will additionally be useful for the proofs of Theorems~\ref{thm:2} and~\ref{thm:3}, which we will give in Section~\ref{sec:hardness-proof} and Section~\ref{sec:approx-algorithm}.

\subsection{2-Approximation by Vertex Separation}

\begin{lemma}[2-Approximation by Vertex Separation] The critical number of corrupt nodes for graph $G$, $m(G)$, 
can be bounded by the minimal sum of $k$-vertex separator and $k$, $\min_k (S_G(k)+k)$, up to a factor of 2. i.e., 
$$\frac{1}{2}\min_{k}{(S_G(k)+k)} \leq m(G)\leq \min_k {(S_G(k)+k)}$$
\label{lemma_2_approx}
\end{lemma}

\begin{proof}[Proof of Lemma \ref{lemma_2_approx}]
 The direction $m(G) \leq \min_k S_G(k)+k$ follows simply. Let $k^* = \arg\min_k (S_G(k)+k)$. If the corrupt party is given $S_G(k^*)+k^*$ nodes to corrupt on the graph, it can first assign $S_G(k^*)$ nodes to the separator, thus the remaining nodes are partitioned into components of size at most $k^*$. Then it arbitrarily assigns one of the components to be all bad nodes. The bad nodes in the connected components report the nodes in the same component as good, and report any node in the separator as bad. The nodes in the separator can effectively report however they want (e.g. report all neighboring nodes as bad). It is impossible to identify even one single good node, because all connected components of size $k$ can potentially be bad, and all vertices in the separator are bad.

The direction $1/2 \min_k (S_G(k)+k)\leq m(G)$ can be proved as follows. When there are $b=m(G)$ corrupt nodes distributed optimally in $G$, it is impossible to find a single good node by definition, and therefore, in particular, the following algorithm (Algorithm~\ref{alg1}) cannot always find a good node:


\begin{algorithm}[H]
Input: Undirected graph $G$
\begin{itemize}
\item If the reports on edge $(u,v)$ does not equal to $(u\in T, v \in T)$, remove both $u,v$ and any incident edges. Remove a pair of nodes in each round, until there are no bad reports left.
\item Call the remaining graph $H$. Declare the largest component of $H$ as good.
\end{itemize}
\caption{Finding one truthful vertex on undirected graph $G$}
\label{alg1}
\end{algorithm}

Run Algorithm \ref{alg1} on $G$, and suppose the first step terminates in $i$ rounds, then: 
\begin{itemize}
\item No remaining node reports neighbors as corrupt
\item $|V|-2i$ nodes remain in graph
\item $\leq b-i$ bad nodes remain in the graph, because each time we remove an edge with bad report, and one of the end points must be a corrupt vertex.
\end{itemize} 

Note that if two nodes report each other as good, they must be the same type (either both truthful, or both corrupt.) Since graph $H$ only contains good reports, nodes within a connected component of $H$ have the same types. If there exists a component of size larger than $b-i$, it exceeds bad party's budget, and must be all good. Therefore, Algorithm~\ref{alg1} would successfully find a good node.

Since Algorithm \ref{alg1} cannot find a good node, the bad party must have the budget to corrupt the largest component of $H$, which means it has size at most $b-i$. Hence, $S_G(b-i)\leq 2i.$ Plugging in $b=m(G),$ we get that 



$$m(G) = \frac{2i}{2}+ b-i\geq \min_k (S_G(k) / 2 + k) \geq \frac{1}{2}\min_k (S_G(k)+k),$$

where the first inequality comes from $2i\geq S_G(b-i).$
\end{proof}


Furthermore, the upperbound in Lemma \ref{lemma_2_approx} additionally tells us that if corrupt party's budget $b\leq m(G)/2$, the set output by Algorithm \ref{alg1} is guaranteed to be good. 

\begin{theorem}
\label{thm:1}
Fix a graph $G$ and suppose that the corrupt party has a budget $b \leq m(G) / 2$. Then the central agency can identify a truthful node, regardless of the strategy of the corrupt party, and without knowledge of either $m(G)$ or $b$. Furthermore, the central agency's algorithm runs in linear time (in the number of edges in the graph $G$).
\end{theorem}

\begin{proof}[Proof of Theorem \ref{thm:1}]
Suppose the corrupt party has budget $b\leq m(G)/2$. Run Algorithm \ref{alg1}. We remove $2i$ nodes in the first step, and separate the remaining graph $H$ into connected components. Notice each time we remove an edge with bad report, at least one of the end point is a corrupt vertex. So we have removed at most $2b \leq m(G)\leq\lceil|V|/2\rceil$ nodes. Therefore, the graph $H$ is nonempty, and the nodes in any connected component of $H$ have the same identity. Let $k^*\geq 1$ be the size of the maximum connected component of $H$. We can conclude that $S_G(k^*)\leq 2i$, since $2i$ is a possible size of $k^*$-vertex separator of $G$. 

Notice there are at most $b-i\leq m(G)/2-i$ bad nodes in $H$ by the same fact that at least one bad node is removed each round. By the upper bound in Lemma \ref{lemma_2_approx},  

$$b-i\leq m(G)/2-i\leq \min_k (S_G(k)+k)/2 -i \leq (2i+k^*)/2-i\leq \frac{k^*}{2}.$$

Since $k^*\geq 1,$ the connected component of size $k^*$ exceeds the bad party's remaining budget $k^*/2$, and must be all good. 

Algorithm \ref{alg1} is linear time because it loops over all edges and removes any ``bad'' edge that does not have reports $(T,T)$ (takes $\leq |E|$ time when we use a list with ``bad'' edges at the front), and counts the size of the remaining components ($\leq |V|$ time), and thus is linear in $|E|$. 
\end{proof}
\begin{remark}
Both bounds in Lemma \ref{lemma_2_approx} are tight. For the lower bound, consider a complete graph with an even number of nodes. For the upper bound, consider a complete bipartite graph with one side smaller than the other.
\label{tight}
\end{remark}

To elaborate on Remark \ref{tight}, for the lower bound, in a complete graph with $n$ nodes, the critical number of bad nodes is $n/2$, and $\min_k S_G(k)+k = n$. 

For the upper bound, consider a complete bipartite graph $G=(V,E)$. The vertex set is partitioned into two sets $V=S_1\cup S_2$ where the induced subgraphs on $S_1$ and $S_2$ consist of isolated vertices, and every vertex $u\in S_1$ is connected with every vertex $v\in S_2$. The smallest sum of $k$-vertex separator with $k$ is obtained with $k = 1$, i.e., $\min_k S_G(k)+k = \min\{|S_1|,|S_2|\}+ 1.$ We argue that this is also the minimal number of bad nodes needed to corrupt the graph. Without loss of generality , let $|S_1|<|S_2|.$ If the bad party controls all of $S_1$ plus one node in $S_2$, it can prevent the identification of a good node. On the other hand, if the bad party controls $b<|S_1|+1$ nodes, then we can always identify a good node. Specifically, we are in one of the following cases: 
\begin{enumerate}
\item The bad party does not control all of $S_1$. Then there will be a connected component of size $n-b > b$ that report each other as good, because the bad nodes cannot control all of $S_2$, and any induced subgraph of a complete bipartite graph with nodes on both sides is connected.
\item The bad party controls all of $S_1$. In this case, the largest connected component of nodes that all report each other as good is only 1. However, in this case, we conclude that the bad nodes must control all of $S_1$ and no other node (due to their budget). Hence, any node in $S_2$ is good.
\end{enumerate}

We end by discussing that the efficient algorithm given in this section does not address the regime when the budget of the bad party, $b$, falls in $m(G)/2 < b \leq m(G)$. Though by definition of $m(G)$, the central agency can find at least one truthful node as long as $b \leq m(G)$, by, for example, enumerating all possible assignments of good/bad nodes consistent with the report, and check the intersection of the assignment of good nodes. However, it is not clear that the central agency has a polynomial time algorithm for doing this. 
Of course, one can always run Algorithm \ref{alg1}, check whether the output set exceeds $b-i/2$, and concludes that the output set is truthful if that is the case. However, there is no guarantee that the output set will be larger than $b-i/2$ if $m(G)/2<b\leq m(G).$ We propose the following conjecture:

\begin{conjecture}
\label{conjecture:1}
Fix a graph $G$ and suppose that the corrupt party has a budget $b$ such that $m(G)/2 < b \leq m(G)$. The problem of finding one truthful node given the graph $G$, bad party's budget $b$ and the reports is NP-hard.
\end{conjecture}

 

%% file: hardness.tex
\subsection{SSE-Hardness of Approximation for $m(G)$}
\label{sec:hardness-proof}
In this section, we show the hardness of approximation result for $m(G)$ within any constant factor under the Small Set Expansion (SSE) Hypothesis \cite{RS10}.
Specifically, we prove Theorem~\ref{thm:2}.

\begin{theorem}
\label{thm:2}
For every $\beta > 1$, there is a constant $\epsilon > 0$ such that the following is true. Given a graph $G = (V,E)$, it is SSE-hard to distinguish between the case where $m(G) \leq \epsilon \cdot |V|$ and $m(G) \geq \beta \cdot \epsilon \cdot |V|$. Or in other words, the problem of approximating the critical number of corrupt nodes for a graph to within any constant factor is SSE-hard.
\end{theorem}


In order to prove Theorem \ref{thm:2},
we construct a reduction similar to \cite{APW12}, and show that the bad party can control auxiliary graph of the \textbf{Yes} case of SSE with $b=O(\epsilon|V'|)$ and cannot control the auxiliary graph of the \textbf{No} case of SSE with $b=\Omega(\epsilon^{0.51}|V'|)$. 


Given an undirected $d$-regular graph $G=(V,E)$, construct an auxiliary undirected graph $G'=(V',E')$ in the following way \cite{APW12}. Let $r = d/2$. For each vertex $v^i\in V$, make $r$ copies of $v^i$ and add to the vertex set of $G'$, denoted $v^i_1,\cdots, v^i_r$. Denote the resulting set of vertices as $\tilde{V}=V\times \{1,\cdots,r\}$. Each edge $e^k \in E$ of $G$ becomes a vertex in $G'$, denoted $e^k$. Denote this set of vertices as $\tilde{E}$. In other words, $V'=\tilde{V}\cup \tilde{E} = V\times \{1,\cdots,r\} \cup E$. There exists an edge between a vertex $v^i_j$ and a vertex $e^k$ of $G'$ if $v^i$ and $e^k$ were adjacent edge and vertex pair in $G$. Note that $G'$ is a bipartite $d$-regular graph with $d/2|V|+|E|=2|E|$ vertices. 


\begin{lemma}
Suppose $q=1/\epsilon$, and $G$ can be partitioned into $q$ equi-sized sets $S_1,\cdots,S_q$ such that $\Phi_G(S_i)\leq 2\epsilon$ for every $1\leq i\leq q.$ Then the bad party can control the auxiliary graph $G'$ with at most $4\epsilon |E| = 2 \epsilon|V'|$ nodes. 
\label{yeslemma}
\end{lemma}
\begin{proof}[Proof of Lemma \ref{yeslemma}]
Notice by Remark \ref{rmk1}, the total number of edges in $G$ not contained in one of the $S_i$ is at most $2\epsilon|E|.$ 

This implies that a strategy for the bad party to control graph $G'$ is as follows. Control vertex $e^k \in \tilde{E}$ if $e^k\in E$ is not contained in any of the $S_i$s in $G$. Call the set of such vertices $E^*\subseteq \tilde{E}$. Let $S_i^*\subseteq V'$ be the set that contains all $r$ copies of nodes in $S_i \subseteq V$. Control one of the $S_i^*s,$ say $S_1^*$. Control all the edge nodes in $\tilde{E}$ that are adjacent to $S_1^*$. Call this set $\mathcal{N}(S_1^*)$. The corrupt nodes in $S_1^*\cup \mathcal{N}(S_1^*)$ report their neighbors in $S_1^*\cup \mathcal{N}(S_1^*)$ as good, and report $E^*$ as bad. Nodes in $E^*$ can effectively report however they want; suppose they report every neighboring node as bad. Then, it is impossible to identify even one truthful node, since assigning any $S_i^*$ as corrupt is consistent with the report and within bad party's budget.

This strategy controls $|E^*|+|S^*_i|+|\mathcal{N}(S^*_1)\setminus E^*|$ nodes on $G'$. Note that $|\mathcal{N}(S^*_1)\setminus E^*|$ is equal to the number of edges that are totally contained in $S_1$ on $G$, which is bounded by $|S_1|\cdot d/2$ (that is if all edges adjacent to $S_1$ are totally contained in $S_1$). If $q=1/\epsilon$, this strategy amounts to controlling $|E^*|+|S^*_i|+|\mathcal{N}(S^*_1)\setminus E^*|\leq 2\epsilon |E|+d/2 \cdot |V|/q + |V|/q\cdot d/2 = 4 \epsilon |E| = 2\epsilon |V'|$ nodes on $G'.$ Notice, this number is guaranteed to be smaller than $1/2|V'|,$ as long as $q>4$.

\label{rmk3}
\end{proof}

Note that, different from the argument in \cite{APW12}, we cannot take $r$ to be arbitrarily large (e.g. $>O(|V||E|)$). This is because when $r$ is large, $2\epsilon |E|+ r \cdot |V|/q = O(\epsilon(|E|+|V'|)) = O(\epsilon|V'|),$ and will not be comparable with the $O(\sqrt{\epsilon}|E|)$ in Lemma~\ref{nolemma}.


\begin{lemma}
Let $G = (V, E)$ be an undirected $d$-regular graph with the property that for every $|V |/10 \leq
|S| \leq 9|V |/10$ we have $|E(S, V \setminus S)| \geq \Omega(\sqrt{\epsilon}|E|)$. If bad party controls $O(\epsilon^{0.51}|E|)=O(\epsilon^{0.51}|V'|) < 1/2|V'|$ nodes on the auxiliary graph $G'$ constructed from $G$, we can always find a truthful node on $G'$. 
\label{nolemma}
\end{lemma}

\begin{proof}[Proof of Lemma \ref{nolemma}]

Assume towards contradiction that the bad party controls $O(\epsilon^{0.51}|E|)$ vertices of graph $G',$ and we can't identify a truthful node. 
\begin{claim}
If the bad party controls $O(\epsilon^{0.51}|E|)$ vertices of graph $G',$ and it is impossible to identify a truthful node, then there exists a set $C$ of size $O(\epsilon^{0.51}|E|)$ and separates $V'\backslash C$ into sets $\{T'_i\}_{i=1,\cdots,\ell}$, each of size $O(\epsilon^{0.51}|E|)$.
\label{sep}
\end{claim}

\begin{proof}[Proof of Claim \ref{sep}]
Since the bad nodes can control $G'$ with $O(\epsilon^{0.51}|E|)$ vertices, 
$m(G') \leq O(\epsilon^{0.51}|E|)$. By the lower bound in Lemma \ref{lemma_2_approx}, $\min_k (S_{G'}(k)+k) \leq 2 m(G') \leq O(\epsilon^{0.51}|E|)$. Let $k^* = \arg \min_k (S_{G'}(k) + k)$. Then $k^*\leq O(\epsilon^{0.51}|E|)$, $S_{G'}(k^*)\leq O(\epsilon^{0.51}|E|)$. By definition of $S_{G'}(k^*)$, there exists a set of size $S_{G'}(k^*)$ whose removal separates the remainder of the graph $G'$ to connected components of size at most $k^*$.
\end{proof}

Let $C$ and $T'_i$ be the sets guaranteed by Claim \ref{sep}. Note we have taken $r = d/2$, and thus $|\tilde{V}| = |\tilde{E}|$. In other words, half of the $V'$ are ``vertex'' vertices $\tilde{V}$, and half are ``edge'' vertices $\tilde{E}$. 
Therefore, with sufficiently small $\epsilon$, $|C\cap \tilde{V}|\leq |C| < 1/2|\tilde{V}|$, $|(\cup_{i=1}^{\ell}T_i')\cap \tilde{V}|\geq 1/2|\tilde{V}|$, $|T_i'\cap \tilde{V}|\leq |T'_i|<3/10|\tilde{V}|$ for every $i$. Therefore, we can merge the different $T_i'$s in Claim \ref{sep}, and have two sets $T_1'$ and $T_2',$ such that $|T_1'\cap \tilde{V}| \geq |\tilde{V}|/5$ and $|T_2'\cap \tilde{V}| \geq |\tilde{V}|/5$. Furthermore, $T_1'$ and $T'_2$ are disjoint, and $T_1', T_2'$, and $C$ cover $V'$. 


Similar to the proof of Lemma 5.1 in \cite{APW12}, we let $T_1\subseteq V$ (resp. $T_2\subseteq V$) be the set of vertices $v\in V$ such that some copy of $v$ appears in $T_1'$ (resp. $T_2'$). Let $S\subseteq V$ be the set of vertices $v\in V$ such that \emph{all} copies of $v$ appear in $C$. Since $|T_1'\cap \tilde{V}|,|T_2' \cap \tilde{V}|\geq |\tilde{V}| / 5= r|V|/5,$ both $|T_1|,|T_2|\geq |V|/5.$ Furthermore, we observe that $T_1 \cup T_2 \cup S = V$, which follows since $T_1' \cup T_2' \cup C = V'$. Now we can lower bound $|T_1 \cup T_2|$ as follows.


$$|T_1\cup T_2| = |V\backslash S| \geq |V|-|C|/r \geq |V| - c\epsilon^{0.51}|E|/r = |V|-c\epsilon^{0.51}|V|,$$
where the first equality uses the fact that $T_1 \cup T_2 \cup S = V$ and that $T_1 \cup T_2$ is disjoint from $S$, and the following inequality uses the fact that $|S| \leq |C| / r$, which follows by definition.


Since $|T_1 \cup T_2|$ is sufficiently large, we can find a balanced partition of $T_1\cup T_2$ into sets $S_1\subseteq T_1$, $S_2\subseteq T_2$, such that $S_1\cap S_2=\emptyset, S_1\cup S_2 = T_1\cup T_2$, and $|V|/10 \leq
|S_1|,|S_2| \leq 9|V|/10$. From the property of $G$ that $E(S,V\setminus S)\geq \Omega(\sqrt{\epsilon}|E|)$ in Lemma \ref{nolemma} and the fact that $G$ is $d$-regular, we know that 

$$E(S_1,S_2) = E(S_1,V \setminus S_1)-E(S_1,S)\geq \alpha \sqrt{\epsilon}|E|-d(\epsilon^{0.51}|E|/r) = \alpha \sqrt{\epsilon}|E|-2\epsilon^{0.51}|E|=\Omega(\sqrt{\epsilon}|E|),$$

for some constant $\alpha$. In the first equality we use the fact that $S_1, S_2, S$ form a partition of $V$. Thus $E(S_1,V\backslash S_1) = E(S_1, S_2\cup S) = E(S_1,S_2) + E(S_1, S).$



Note that since $S_1\subseteq T_1$ and $S_2\subseteq T_2$, and $T_1'$ and $T_2'$ do not have edge between them in $G'$, the edges $E(S_1,S_2)$ all have to land as "edge vertices" in $C$. In other words, for any $u\in S_1$, and $v \in S_2,$ if $(u,v)\in E$, then the vertex $(u,v)\in V'$ has to be included in the set $C$, thus $|C|\geq \Omega(\sqrt{\epsilon}|E|)$.

This contradicts the fact that there are only  $O(\epsilon^{0.51}|E|)$ vertices in $C$. 

\end{proof}

Combining Lemma \ref{yeslemma} and Lemma \ref{nolemma}, Theorem \ref{thm:2} follows in standard fashion. We give a proof here for completeness.


\begin{proof}[Proof of Theorem \ref{thm:2}]
Suppose for contradiction that there exists some constant $\beta > 0$ such that there is polynomial time algorithm $\mathcal{A}$ that does the following. For any $\epsilon' > 0$ and an arbitrary graph $G' = (V',E')$, it can distinguish between the case where $m(G') \leq \epsilon' \cdot |V'|$ and $m(G') \geq \beta \cdot \epsilon' \cdot |V'|$. Specifically, we will suppose this holds for $\epsilon' < \frac{1}{\beta^{2.05}}$. Then we can use this algorithm to decide the SSE problem as follows.

Fix $\epsilon < \frac{1}{1.5 \beta^{2.05}}$, $q = 1 / \epsilon$, $\gamma > 0$ sufficiently small ($\leq o(\sqrt{\epsilon})$ suffices). Let $G = (V,E)$ be an arbitrary input to the resulting instance of the SSE decision problem (from Conjecture~\ref{SSE}). Construct the graph $G' = (V', E')$ from $G$ as done in the beginning of Section~\ref{sec:hardness-proof}. 

If $G$ was from the YES case of Conjecture~\ref{SSE}, then $m(G') \leq 1.5 \epsilon |V'|$ (Lemma~\ref{yeslemma}). If $G$ was from the NO case of Conjecture~\ref{SSE}, then $m(G') > \epsilon^{0.51} |V'|$ (Lemma~\ref{nolemma}). We can invoke our algorithm $\mathcal{A}$ to distinguish these two cases, by letting $\epsilon' = 1.5 \epsilon$ and noting that $\beta < (1/(\epsilon')^{0.49})$ by design, which would decide the problem in Conjecture~\ref{SSE} in polynomial time.
\end{proof}





Now, we can obtain the following Corollary \ref{corr:1} from Theorem \ref{thm:2}.

\begin{corollary}
Assume the SSE Hypothesis and that P $\neq$ NP. Fix any $\beta > 1$. There does not exist a polynomial-time algorithm that takes as input an arbitrary graph $G = (V,E)$ and outputs a set of nodes $S$ with size $|S|\leq O(\beta \cdot m(G))$, such that corrupting $S$ prevents the central agency from finding a truthful node.
\label{corr:1}
\end{corollary}


In summary, the analysis in this section tells us that given an arbitrary graph, it is hard for bad party to corrupt the graph with minimal resources. On the other hand, if the budget of bad nodes is a factor of two less than $m(G)$, a good party can always be detected with an efficient algorithm, e.g. using Algorithm \ref{alg1}.


\subsection{An $O(\log |V|)$ Approximation Algorithm for $m(G)$}
\label{sec:approx-algorithm}
In light of the SSE-hardness of approximation of $m(G)$ within any constant, and the close relation of $m(G)$ with $k$-vertex separator, we leverage the best known approximation result for $k$-vertex separator to propose an $O(\log n)$ approximation algorithm for $m(G)$. It is useful as a test for central authorities for measuring how corruptible a graph is. Notably, it is also a potential algorithm for (computationally restricted) bad party to use to decide which nodes to corrupt.

The paper \cite{Lee17} presents an bicritera approximation algorithm for $k$-vertex separator, with the guarantee that for each $k$, the algorithm finds a subset $B_k \subseteq V$ such that $|B_k|\leq O(\frac{\log k}{\epsilon})\cdot S_G(k)$, and the induced subgraph $G_{V\backslash B_k}$ is divided into connected components each of size at most $k/(1-2\epsilon)$ vertices.

\begin{proposition}[Theorem 1.1, \cite{Lee17}]
For any $\epsilon \in (0,1/2),$ there is a polynomial-time $(\frac{1}{1-2\epsilon},O(\frac {\log k} \epsilon))$- bicriteria approximation algorithm for $k$-vertex separator. 
\label{lee}
\end{proposition}

Interested readers can refer to \cite{Lee17} Section 3 for the description of the algorithm. Leveraging this algorithm for $k$-vertex separator, we can obtain a polynomial-time algorithm for seeding corrupt nodes and preventing the identification of a truthful node.

\begin{theorem}[$O(\log |V|)$ Approximation Algorithm]
There is a polynomial-time algorithm that takes as input a graph $G = (V,E)$ and outputs a set of nodes $S$ with size $|S|\leq O(\log |V|\cdot m(G))$, such that corrupting $S$ prevents the central agency from finding a truthful node.
\label{thm:3}
\end{theorem}

\begin{proof}
The algorithm is as follows. Call the bicriteria algorithm for approximating $k$-vertex separator in \cite{Lee17} $n$ times, once for each $k$ in $k=1,\cdots,n$, where $n = |V|$. Each time the algorithm outputs a set of vertices $B_k$ that divides the remaining graph into connected components with maximum size $g(k)$. Choose the $k^*$ for which the algorithm outputs the smallest value of $\min_k|B_k|+g(k)$. The bad party can control $B_{k^*}$ and one of the remaining connected components (the size of which is at most $g(k^*)$), and be sure to prevent the identification of one good node, by the same argument that lead to the upper bound in Lemma \ref{lemma_2_approx}.

We now prove that $|B_{k^*}|+g(k^*)$ is an $O(\log |V|)$ approximation for the quantity of consideration $\min_k S_G(k)+k$. For each $k$, we denote our approximation for $S_G(k)+k$ as $f(k):=|B_k|+g(k).$ Then by the guarantee given in Proposition \ref{lee}, we know

$$f(k)= |B_k| + g(k) \leq O\left(\frac{\log k}{\epsilon}\right)\cdot S_G(k) + \frac{1}{1-2\epsilon} k\leq O\left(\frac{\log k}{\epsilon}\right)\cdot (S_G(k) + k).$$

Thus $$\min_k f(k) \leq \min_k O\left(\frac{\log k}{\epsilon}\right)\cdot (S_G(k) + k)\leq O\left(\frac{\log n}{\epsilon}\min_k (S_G(k) + k)\right)\leq O\left(\log n \cdot m(G)\right).$$

The last inequality follows from the fact that that $\min_k (S_G(k) + k)/2\leq m(G)\leq \min_k (S_G(k) + k)$ in Lemma~\ref{lemma_2_approx}, and by taking $\epsilon$ to be a fixed constant, e.g. $\epsilon = 1/3$. So $\min_k f(k)$  provides an $O(\log n)$ approximation of $m(G)$. The algorithm consists of $n$ calls of the polynomial-time algorithm in Proposition~\ref{lee}, so is also polynomial-time.
\end{proof}

%% file: directed.tex
\section{Directed Graphs}
Here we present the variant of our problem on directed graphs. As discussed in \cite{AMP15}, this is motivated by the fact that in various auditing
situations, it may not be natural that any $u$ will be able to inspect $v$ whenever $v$ inspects $u$. 

Given a directed graph $D=(V,E_D)$, we are asked to to find $m(D)$, the minimal number of corrupted agents needed to prevent the identification of a single truthful agent. Firstly, since undirected graphs are special cases of directed graphs, it is clear that the worst case hardness of approximation results still hold. In this section, we will define a analogous notion of vertex separator relevant to corruption detection for directed graphs, and state the version of Theorem \ref{thm:1} for directed graphs. 

\begin{definition}[\textbf{Reachability Index}] On a directed graph $D=(V,E_D)$, say a vertex $s$ can reach a vertex $t$ if there exists a sequence of adjacent vertices (i.e. a path) which starts with $s$ and ends with $t$. Let $R_D(v)$ be the set of vertices that can reach a vertex $v$.
Define the \textbf{reachability index} of $v$ as $|R_D(v)|$, or in other words, as the total number of nodes that can reach $v$. 
\end{definition}

Based on the notion of reachability index, we design the following algorithm, Algorithm \ref{alg2}, for detecting one good node on directed graphs:


\begin{algorithm}[H]
Input: Directed graph $D$
\begin{itemize}
\item If node $u$ reports node $v$ as corrupt, remove both $u,v$ and any incident edges (incoming and outgoing). Remove a pair of nodes in each round. Continue until there are no bad reports left.
\item Call the remaining graph $H=(V_H,E_H)$. Declare a vertex in $H$ with maximum reachability index as good.
\end{itemize}
\caption{Finding one truthful vertex on directed graph $D$}
\label{alg2}
\end{algorithm}

Run Algorithm \ref{alg2} on directed graph $D$, and suppose the first step terminates in $i$ rounds. Then: 
\begin{itemize}
\item No remaining node reports out-neighbors as corrupt
\item $|V|-2i$ nodes remain in graph
\item $\leq b-i$ bad nodes remain in the graph, because each round in step 1 removes at least one bad node.
\end{itemize} 

The main idea is that, if there exists a node $v$ with reachability index larger than $b-i$, at least $b-i$ nodes claim (possibly indirectly) that $v$ is good, which means at least one good node also reports $v$ as good, and thus $v$ must be good. In the rest of the section, we use this observation to generalize Theorem \ref{thm:1}.

We define a notion similar to $k$-vertex separator on directed graphs, show that our notion provides a $2$-approximation for $m(D)$ when $D$ is a directed graph, and that the equivalent of Theorem \ref{thm:1} also holds in the directed case.

\begin{definition}[\textbf{$k$-reachability separator}] We say a set of vertices $S\subseteq V$ is a $k$-reachability separator of a directed graph $D=(V,E_D)$ if after the removal of $S$ and any adjacent edges, all vertices in the remaining graph are of reachability at most $k$.
\end{definition}

Since in an undirected graph, any pair of vertices can reach each other if and only if they belong to the same connected component, one can check that $k$-reachability separator on an undirected graph is exactly equivalent to a $k$-vertex separator. Thus we use a similar notation, $S_D(k)$, to denote the size of the minimal $k$-reachability separator on $D$.

\begin{lemma}[2-Approximation Lemma on Directed Graphs]
$$ \frac{1}{2}\min_k (S_{D}(k)+k) \leq m(D) \leq \min_k (S_{D}(k)+k)$$
\end{lemma}

\begin{proof}
The direction $m(D) \leq \min_k S_D(k)+k$ is proved as follows. Let $k^* = \arg\min_k (S_D(k)+k)$. If the corrupt party is given $\min_k (S_D(k)+k)$ nodes to allocate on $D$, it can first assign $S_D(k^*)$ nodes to a $k^*$-reachability separator $C$, such that the remaining nodes have reachability index at most $k^*$. Then it arbitrarily assigns one of the vertices $v^*$ with maximum reachability index plus its $R_H(v^*)$ as bad. The bad nodes in $R_H(v^*)$ report any neighbor in the separator $C$ as bad and any other neighbor as good. The nodes in the separator can effectively report however they want (e.g. report all neighboring nodes as bad). 

It is impossible to detect a single good node, because every node $v$ can only be reached by $R_H(v)$ and $C$. For every $v\in H$, it being assigned as corrupt or good is consistent with the reports. If $v$ is corrupt, $R_H(v)$ is also assigned as corrupt, thus all nodes in $H$ receive good reports from $R_H(v)$, bad reports from $C$ and give bad reports to $C$. If $v$ is truthful, all nodes still receive and give the same reports. So for every $v\in V_H$, assigning $R_H(v)$ as bad, and $V_H\setminus R_H(v)$ as good is consistent with the observed reports. It is impossible to find a good node in $H$ by definition.

The proof for $1/2 \min_k (S_{D}(k)+k) \leq m(D)$ is given by Algorithm \ref{alg2}. Let there be $m(D)$ bad nodes distributed optimally on the graph. By definition, these nodes prevent the identification of a good node. Run Algorithm~\ref{alg2}, and suppose the first step terminates in $i$ rounds. This means we have removed at least $i$ bad nodes, and there are at most $m(D)-i$ bad nodes left on $H$. If there exists a node on $H$ with reachability $m(D)-i$, then this node must be truthful, since there are not enough bad nodes left to corrupt all the nodes that can reach it, and all the reports in the remaining graph are good. Thus $|R(v)| < m(D)-i$ for any $v$. Therefore, the set of $2i$ removed nodes must be an $m(D)-i$ reachability separator. Hence, we can bound $m(D)$ as follows.

$$m(D)= (m(D)-i)+2i/2 \geq \min_k (k+S_D(k)/2) \geq \frac{1}{2} \min_k (S_D(k)+k)$$
where the first inequality follows from the fact that $2i \geq S_D(m(D) - i)$. 
\end{proof}

\begin{theorem}
\label{thm:4}
Fix a directed graph $D$ and suppose that the corrupt party has a budget $b \leq m(D) / 2$. Then the central agency can identify a truthful node, regardless of the strategy of the corrupt party, and without the knowledge of either $m(D)$ or $b$. Furthermore, the central agency's algorithm runs in linear time.
\end{theorem}

\begin{proof}[Proof of Theorem \ref{thm:4}]
Suppose the corrupt party has budget $b\leq m(D)/2$. Run Algorithm \ref{alg2}. Notice each time we remove an edge with bad report, at least one of the end point is a corrupt vertex. So we have removed at most $2b \leq m(D)\leq\lceil|V|/2\rceil$ nodes. Therefore, the graph $H$ is nonempty. Let $k^*\geq 1$ be the maximum reachability index in $H$. Since $b\leq m(D)/2$, and there are no bad reports in $H$, the reachability index of a bad node in graph $H$ is at most $m(D)/2-i\leq \min_k (S_{D}(k)+k)/2-i \leq (2i+k^*)/2-i = k^*/2 < k^*$. 

Then a vertex with reachability index $k^*$ must be found by Algorithm \ref{alg2}, and must be a truthful node. The linear runtime $O(|E_D|)$ follows from the same analysis as in the proof of Theorem~\ref{thm:1}.
\end{proof}




\section{Finding an Arbitrary Fraction of Good Nodes on a Graph} 
\label{sec:find-many-truthful}
Being able to detect one good node may seem limited, but in fact, the same arguments and construction can be adapted to show that approximating the critical number of bad nodes to prevent detection of any arbitrary $\delta$ fraction of good nodes is SSE-hard. In this section, we propose the definition of $g$-remainder $k$-vertex separator, a vertex separator notion related to identifying arbitrary number of good nodes, present a $2$-approximation result, and prove hardness of approximation with arguments similar to proof of Theorem \ref{thm:2} in Section~\ref{sec:hardness-proof}. 

We abuse notation and define $m(G,g)$ to be the minimal number of bad nodes needed to prevent the identification of $g$ nodes. 
\begin{definition}[$m(G,g)$]
We define $m(G,g)$ as the minimal number of bad nodes such that it is impossible to find $g$ good nodes in $G$. In particular, $m(G) = m(G,1)$. 
\end{definition}

\begin{definition}[\textbf{$g$-remainder $k$-vertex Separator}]
Consider the following separation property: after the removal of a vertex set $S$, the remaining graph $G_{V\backslash S}$ is a union of connected components, where connected components of size larger than $k$ sum up to size less than $g$. We call such a set $S$ a $g$-remainder $k$-vertex separator of $G$. 

For any integer $0<k,g<|V|$, we denote the minimal size of such a set as $S_G(k,g).$ In particular, a minimal $k$-vertex separator is a $1$-remainder $k$-vertex separator, i.e., $S_G(k) = S_G(k,1).$
\end{definition}

\begin{theorem}
Fix a graph $G$ and the number of good nodes to recover, $g$. Suppose that the corrupt party has a budget $b\leq m(G,g) / 2$. If $g<|V|-2b,$ then the central agency can identify $g$ truthful nodes, regardless of the strategy of the corrupt party, and without knowledge either of $m(G,g)$ or $b$. Furthermore, the central agency's algorithm runs in linear time. 
\label{thm:efficient_alg}
\end{theorem}

\begin{algorithm}[H]
Input: Undirected graph $G$
\begin{itemize}
\item If the reports on edge $(u,v)$ does not equal to $(u\in T, v \in T)$, remove both $u,v$ and any incident edges. Remove a pair of nodes in each round, until there are no bad reports left.
\item Suppose the previous step terminates in $i$ rounds. In the remaining graph $H$, rank the connected component from large to small by size. Declare the largest component as good and remove the declared component until we have declared $g$ nodes as good.
\end{itemize}
\caption{Finding $g$ truthful vertices on an undirected graph $G$}
\label{alg3}
\end{algorithm}

\begin{proof}[Proof of Theorem \ref{thm:efficient_alg}]
We claim that central agency can use Algorithm \ref{alg3}, and output at least $g$ good nodes if $b\leq m(G,g)/2$. Step 1 of Algorithm \ref{alg1} must terminate after removing fewer than $m(G,g)$ nodes, because each round has to remove at least one bad node, and there are only $m(G,g)/2$ bad nodes in total. Let the number of nodes removed be $m(G,g)-\delta$, so at least $m(G,g)/2 - \delta/2 \geq b-\delta/2$ are corrupt. Thus at most $\delta/2$ bad nodes remain in the graph $H$. 

Assume towards contradiction that only $y<g$ nodes output by Algorithm \ref{alg1} are good. This means that the $m(G,g)-\delta$ removed nodes separate the graph $G$ into connected components where all components with size larger than $\delta / 2$ sum to fewer than $g$. Then $m(G,g)-\delta = m(G,y)$ for $y < g$, contradicting the fact that $m(G,g)$ is the minimum budget needed to prevent identification of $g$ nodes. 









\end{proof}

In fact, just like in Section~\ref{sec:main-results}, Algorithm \ref{alg3} additionally gives us a characterization of $m(G,g)$ in terms of the size of the smallest $g$-remainder $k$-vertex separator of a graph, for an appropriately chosen value of $k$.

\begin{lemma}[2-Approximation by Vertex Separation]
The minimal sum of $g$-remainder $k$-vertex separator and $k$, $\min_k{(S_G(k,g)+k)}$, bounds the critical number of bad nodes $m(G,g)$ up to a factor of 2. i.e., 
$$\frac{1}{2}\min_{k}{S_G(k,g)+k} \leq m(G,g)\leq \min_k {S_G(k,g)+k}.$$
\label{lemma_2_approx2}
\end{lemma}

\begin{proof}[Proof of Lemma \ref{lemma_2_approx2}]
The upper bound follows simply. Let $k^* = \arg \min_k S_G(k,g)+k$. Given a budget $b=\min_kS_G(k,g)+k$, the bad party can remove a set of size $S_G(k^*,g)$ and separate the graph into connected components of size at most $k^*$, except for fewer than $g$ nodes. Control one of the connected components of size at most $k^*$, and construct the reports similarly as in Lemma \ref{lemma_2_approx}. Then the central agency can only identify fewer than $g$ good nodes. 

For the lower bound, suppose there are $b=m(G,g)$ bad nodes distributed optimally on $G$ and thus it's impossible to find $g$ good nodes by definition. Run Algorithm \ref{alg3}. Suppose the first step terminates in $i$ rounds. After the removal of $2i$ nodes, the graph must be separated into connected components smaller than $b-i$, except for fewer than $g$ nodes. Then $2i \geq S_G(b-i,g)$. Therefore, 

$$ \frac{1}{2} \min_k \left( S_G(k,g)+ k \right) \leq \min_k \left( \frac{S_G(k,g)}{2} + k \right) \leq \frac{1}{2}S_G(b-i,g)+ (b-i) \leq \frac{2i}{2} +b-i = m(G,g) $$
\end{proof}

Now using the characterization given by $g$-remainder $k$-vertex separator, we are ready to prove that it is SSE-hard to approximate the budget needed to prevent any arbitrary number of good nodes, i.e., $m(G,g)$ for any $g<|V|/3$. 

\begin{theorem}
For every $\beta>1$ and every $0<\delta<1$, there is a constant $\epsilon>0$ such that the following is true. Given a graph $G=(V,E)$, it is SSE-hard to distinguish between the case where $m(G,\delta |V|)\leq \epsilon \cdot |V|$ and $m(G,\delta|V|)\geq \beta \cdot \epsilon \cdot |V|.$ Or in other words, the problem of approximating the critical number of corrupt nodes such that it is impossible to find $\delta |V|$ good nodes within any constant factor is SSE-hard.
\label{thm:sse_hard}
\end{theorem}

We first prove Theorem \ref{thm:sse_hard} for $0<\delta<1/3$. The proof in this regime follows similar constructions and arguments as in the proof of Theorem \ref{thm:2}. Note that the proof extends naturally for any $0<\delta<1/2$. This is effectively because the range for $\mu$ in Remark \ref{rmk2} can be made to $[\epsilon',1/2]$, for any constant $\epsilon'>0$. Further explanation is provided in proof for Lemma \ref{lemma:no}. 

Firstly, we construct $G'$ based on $G$ as in Section~\ref{sec:hardness-proof}. Lemma \ref{yeslemma} immediately implies that:

\begin{lemma}
Suppose $q=1/\epsilon$, and $G$ can be partitioned into $q$ equi-sized sets $S_1,\cdots,S_q$ such that $\Phi_G(S_i)\leq 2\epsilon$ for every $1\leq i\leq q.$ The bad party can prevent the identification of one good node, and thus $\delta |V'|$ good nodes, on the auxiliary graph $G'$ with $O(\epsilon |E|)=O(\epsilon|V'|)$ nodes.
\label{lemma:yes}
\end{lemma}

We reprove the analogous lemma to Lemma \ref{nolemma}.
\begin{lemma}
Let $G = (V, E)$ be an undirected $d$-regular graph with the property that for every $|V |/10 \leq
|S| \leq 9|V |/10$ we have $|E(S, V \ S)| \geq \Omega(\sqrt{\epsilon}|E|)$. If bad party controls $O(\epsilon^{0.51}|E|)=O(\epsilon^{0.51}|V'|) < 1/2|V'|$ nodes on the auxiliary graph $G'$ constructed from $G$, we can always find $\delta|V'|$ truthful nodes on $G'$, for any $\delta<1/3$. 
\label{lemma:no}
\end{lemma}

\begin{proof}[Proof of Lemma \ref{lemma:no}]
Let $g=\delta |V'|$. 
Assume towards contradiction that the bad party controls $O(\epsilon^{0.51}|E|)$ vertices in $G',$ and we cannot identify $g$ truthful nodes. 

\begin{claim}If the bad party controls $O(\epsilon^{0.51}|E|)$ vertices of graph $G',$ and we can't identify $g$ truthful node, then there exists a set $C$ of size $O(\epsilon^{0.51}|E|)$ and separates $V'\backslash C$ into sets $\{T'_i\}_{i=1,\cdots,\ell}$, each of size $|T'_i|\leq O(\epsilon^{0.51}|E|)$, and sets $\{A'_j\}_{j=1,\cdots,K}$, each of size $|A'_j|>\Omega(\epsilon^{0.51}|E|)$, and $|\cup_j^K A'_j|< g$.
\label{sep2}
\end{claim}

\begin{proof}[Proof of Claim \ref{sep2}]
Since the corrupt party can control $G'$ with $O(\epsilon^{0.51}|E|)$ vertices, 
$m(G',g) \leq O(\epsilon^{0.51}|E|)$. By Lemma \ref{lemma_2_approx2} $\min_k S_{G'}(k,g)+k \leq 2 m(G',g) \leq O(\epsilon^{0.51}|E|)$. Let $k^* = \arg \min_k S_{G'}(k,g) + k$. Then $k^*\leq O(\epsilon^{0.51}|E|)$, $S_{G'}(k^*,g)\leq O(\epsilon^{0.51}|E|)$. By definition of $ S_{G'}(k^*,g)$, there exists a set of size $S_G(k^*)$ after whose removal separates the remainder of the graph $G$ to connected components of size at most $k^*$ except for fewer than $g$ nodes. Thus components of size larger than $\Omega(\epsilon^{0.51}|E|)$ contain fewer than $g$ nodes.
\end{proof}

Let $T'= \cup_{i=1}^\ell T'_i, A'= \cup_{j=1}^K A'_j.$ Since $|C|=O(\epsilon^{0.51}|E|)=O(\epsilon^{0.51}|\tilde{V}|)$, and $C\cup T'\cup A'=V'$, for small enough $\epsilon$, $|(T'\cup A') \cap \tilde{V}|\geq 9|\tilde{V}|/10.$ From the assumption that we can't identify $g$ truthful nodes, $|A'| < g \leq |V'|/3$. Otherwise, we can claim the entire $A'$ as good and identify $g$ truthful nodes. Thus $|A'\cap \tilde{V}|\leq |V'|/3 \leq 2/3|\tilde{V}|$.\footnote{If we use the fact that $|A'|<g\leq (|V'|-\epsilon'|V|)/2$, for some constant $\epsilon'$, then $|A'\cap \tilde{V}|\leq (|V'|-\epsilon'|V'|)/2\leq (1-\epsilon')|\tilde{V}|$. We can merge $\{\{A_j\}, \{T_i\}\}$ to two sets $V'_1$, $V'_2$ such that $|V'_1\cap \tilde{V}|,|V'_2\cap \tilde{V}| \geq \epsilon'|\tilde{V}|.$ The rest of the proof still goes through.} 

Additionally, use the fact that
$|T'_i\cap \tilde{V}|<|\tilde{V}|/10$ for every $i$, with sufficiently small $\epsilon,$ we can merge various sets in $\{\{A_j\}_{j=1,\cdots,K},\{T_i\}_{i=1,\cdots, \ell}\}$ and get two sets $V_1'$ and $V_2',$ such that $|V'_1\cap \tilde{V}|,|V'_2\cap \tilde{V}|\geq |\tilde{V}|/10$, and $V_1'$ and $V_2'$ are separated by $C$. 

Now, let $V_1\subseteq V$ (resp. $V_2\subseteq V$) be the set of vertices $v\in V$ such that some copy of $v$ appears in $V_1'$ (resp. $V_2'$). Let $S\subseteq V$ be the set of vertices $v\in V$ such that all $r$ copies of $v$ appears in $C$. Since $|V_1'\cap \tilde{V}|,|V_2'\cup \tilde{V}|\geq |\tilde{V}|/10=r|V|/10,$ both $|V_1|,|V_2|\geq |V|/10.$ Furthermore, we observe that $V_1\cup V_2\cup S = V$, which follows from $V_1'\cup V_2'\cup C =V'.$
Now we can lower bound $|V_1\cup V_2|$ as follows. 

\begin{equation*}
|V_1\cup V_2| = |V\backslash S| \geq |V|-|C|/r \geq |V| - c\epsilon^{0.51}|E|/r = |V|-c\epsilon^{0.51}|V| \\
\end{equation*}
The first equality again follows from the fact that $V_1\cup V_2\cup S= V$, and that $V_1\cup V_2$ is disjoint from $S$, and the second inequality follows by definition of $S$.

Since $V_1\cup V_2$ is sufficiently large, we can find a balanced partition of $V_1\cup V_2$ into sets $S_1\subseteq V_1$, $S_2\subseteq V_2$, $S_1\cap S_2=\emptyset, S_1\cup S_2 = V_1\cup V_2$, $|V|/10 \leq
|S_1|,|S_2| \leq 9|V|/10$. From the property of $G$ that $E(S,V\setminus S)\geq \Omega(\sqrt{\epsilon}|E|)$ in Lemma \ref{nolemma} and the fact that $G$ is $d$-regular, we know that 

$$E(S_1,S_2) = E(S_1,V \setminus S_1)-E(S_1,S)\geq \alpha \sqrt{\epsilon}|E|-d(\epsilon^{0.51}|E|/r) = \alpha \sqrt{\epsilon}|E|-2\epsilon^{0.51}|E|=\Omega(\sqrt{\epsilon}|E|),$$
for some constant $\alpha$. In the first equality, we use the fact that $S_1 \cup S_2 \cup S = V$, and $S_1,S_2,S$ are disjoint. Thus $E(S_1, V\backslash S_1) = E(S_1, S_2\cup S) = E(S_1,S_2)+E(S_1,S).$

Note that since $S_1\subseteq V_1$ and $S_2\subseteq V_2$, and $T_1'$ and $T_2'$ do not have edges between them in $G'$, the edges $E(S_1,S_2)$ all have to land as "edge vertices" in $C$. Formally, $E(S_1,S_2)\subseteq \tilde{E}\cap C$. In other words, for any $u\in S_1$, and $v \in S_2,$ if $(u,v)\in E$, then the vertex $(u,v)\in V'$ has to be included in the set $C$, thus $|C|\geq \Omega(\sqrt{\epsilon}|E|)$.

This contradicts the fact that there are only $O(\epsilon^{0.51}|E|)$ vertices in $C$. 
\end{proof}

Using Lemma \ref{lemma:yes} and Lemma \ref{lemma:no}, we can again obtain Theorem \ref{thm:sse_hard} for $0<\delta<1/2$, with the same argument for the proof of Theorem \ref{thm:2} in Section \ref{sec:hardness-proof}. 

When $1/2\leq\delta<1$, we construct an auxiliary graph in the following way. Take as input any graph $G=(V,E)$. Let $h=\delta /(1-\delta)|V|$, construct $G' = G\cup h$-clique. Note $h=\delta|V'|.$ Then, we claim that the critical number of bad nodes such that it is impossible to detect $\delta|V'|+1$ good nodes on $G'$ is the same as the critical number of bad nodes such that it is impossible to find one good node on $G$.

\begin{claim}
Given any graph $G$, $1/2\leq \delta <1$ and $G'$ as constructed,
$$m(G',\delta |V'|+1) = m(G).$$
\label{claim:3}
\end{claim}

\begin{proof}
Firstly, observe that $$\delta |V'|  = \delta(|V|+h) = \delta (|V| + \delta/(1-\delta)|V|)= h.$$ Therefore, one way to prevent identification of  $\delta |V'|+1$ good nodes on $G'$ is to prevent identification of one good node on $G$. Since the $h$-clique is of size at least $|V'|/2$, and report each other as good, they will be detected as good nodes. This strategy requires bad party to have budget $b=m(G).$ Thus $m(G',\delta|V'| + 1)\leq m(G)$.

The direction $m(G',\delta|V'| + 1)\geq m(G)$ follows by the fact that the strategy above is optimal. In order to prove this, we make the following observation: 

\begin{claim}
Given any graph $G$ and $g\leq |V|,$
$$m(G) \leq m(G,g)+g-1$$
\label{claim:4}
\end{claim}

\begin{proof}[Proof of Claim \ref{claim:4}]
One way to prevent identification of one good node is to corrupt $m(G,g)$ nodes plus the (at most) $g-1$ detected good nodes. Call the set of the $g-1$ or fewer detected nodes $S$. Notice that any node in $G\backslash S$ that is adjacent to $S$ are reported as bad by $S$. If not, this node has the same identity with $S$, and should be detected as good as well. Therefore, the bad party is able to corrupt the set $S$ without incurring any change in the reports, since all edges incident to $S$ now have both endpoints corrupt and so the reports are arbitrary. Previously, the set $S$ were good in any configuration of identities consistent with the reports and the budget. But now, the bad party's budget increases by at least $g-1 \geq |S|$, and any configuration with the set $S$'s identity changed to all bad is also consistent with the reports. 

Therefore, no node is good in all configurations, and so no node can be detected as good. This strategy requires $m(G,g)+g-1$ nodes and prevents identification of one good node. Since $m(G)$ is the minimal number of bad nodes so that it is impossible to detect one good node, $m(G)\leq m(G,g)+g-1$.
\end{proof}

Now we continue to prove the $m(G',\delta|V'|+1)\geq m(G)$ direction of Claim \ref{claim:3}. Assume towards contradiction that there exists a strategy that controls at least one node in the $h$-clique, prevents identification of $h+1$ good nodes, and requires fewer than $m(G)$ bad nodes in total. Suppose this strategy assigns $a$ nodes in the $h$-clique as bad, where $1<a < m(G) \leq |V|/2\leq h/2$. Then $h-a > h/2 > m(G) >b$. Therefore, the rest of the $h$-clique forms a connected component with only good reports, and is of size $h-a$, which is larger than the bad party's budget $b<m(G)$, thus are declared as good. As a result, the bad party must prevent identification of $a+1$ good nodes in $G$ with budget strictly less than $m(G)-a$. This contradicts the fact that $m(G)-a \leq m(G,a) - 1 < m(G,a+1)$ by Claim~\ref{claim:4}.  

Therefore, the strategy of controlling $m(G)$ nodes on $G$ and let the $h$-clique be detected as good is an optimal strategy, $m(G',\delta|V'|+1) = m(G)$.
\end{proof}

Now, with Claim \ref{claim:3}, we conclude that for any $1/2\leq\delta<1$, approximating $m(G,\delta |V|)$ within any constant must be SSE-hard. If not, we will obtain an efficient algorithm for approximating $m(G)$ by constructing a graph $G'$ by adding a $\frac{\delta}{1-\delta}|V|$ clique to any graph $G$, for some $\delta$, and approximate $m(G)$ by approximating $m(G',\delta|V'|+1),$ which is just $m(G',\delta'|V'|)$ for some other $0<\delta'<1$.

Theorem \ref{thm:sse_hard} implies a similar corollary about the SSE-hardness of seeding the nodes on a graph $G$ given any constant multiple of the critical number $m(G,\delta|V|)$ to prevent detection of any arbitrary fraction of good nodes.

\begin{corollary}
Assume the SSE Hypothesis and P $\neq$ NP. Fix any $\beta > 1,$ and $0<\delta<1$. There does not exist a polynomial-time algorithm that takes as input an arbitrary graph $G = (V,E)$ and outputs a set of nodes $S$ with size $|S|\leq O(\beta \cdot m(G,\delta|V|))$, such that corrupting $S$ prevents the central agency from finding $\delta |V|$ truthful nodes.
\label{corr:2}
\end{corollary}


%% file: appendix.tex
\section{Omitted Results}
We give an NP-hardness result for computing $\min_{k}{S_G(k)+k}$ exactly. Note that this is insufficient to say anything about corruption detection, as $\min_{k}{S_G(k)+k}$ only gives a 2-approximation to the critical number $m(G)$, but we include this observation here as it may be of independent interest.
\begin{theorem}
It is NP hard to compute $\min_{k}{S_G(k)+k}$ exactly.
\label{thm_hardness}
\end{theorem}
\begin{proof}
It is known that finding $k$-vertex separator for a graph is NP hard \cite{Lee17}. We present a reduction of the problem of computing $\min_k{S_G(k)+k}$ to the $k$-vertex separator problem.

Assume towards contradiction that there is a polynomial-time algorithm $\mathcal{A}$ for finding $\min_k S_G(k)+k$. Then for any graph $G$ and any $M<|V|$, the minimal $M$-vertex separator of the graph $G=(V,E)$ can be found in the following way. Construct a graph $G'=(V',E'),$ where $$G'=G\cup \{n^2 \textrm{ disjoint M-cliques}\} ,$$
with $n \gg N:=|V|.$
Construct a second auxiliary graph $G''=(V'',E''),$ such that

$$G''=G'\cup \{kn+N \textrm{ disjoint } (n-1)\textrm{-cliques appended to each vertex of V'}\}. $$

Each $(n-1)$-clique is appended to a vertex of $G'$ in the sense that each node of the clique is connected to the vertex in $G'$ with an edge. The idea is to make each vertex in $G'$ "$n$ times larger".

Run the polynomial-time algorithm $\mathcal{A}$ for finding $\min_k S_{G''}(k)+k$ on graph $G''$. The algorithm outputs a vertex set $S'' \subseteq V'',$ which divides $G''$ into connected components of with maximal size $k''$. 

\begin{lemma}
Let $G''$ be as constructed above, $k''$ and $S''$ be the output given by an algorithm that computes $\min_{k}S_{G''}(k)+k$. Then $k''=nM$, and without loss of generality, the subset $S''$ contains only vertices from the original graph $G$. In other words, finding $\min_k S(k)+k$ of $G''$ is equivalent to finding the $M$-vertex separator of $G$. i.e.,
$$\arg\min_k  S_{G''}(k)+k = nM, $$ $$\min_k S_{G''}(k)+k = S_G(M)+nM.$$
\label{claim_eq}
\end{lemma}

\begin{proof} [Proof of Lemma \ref{claim_eq}]
Let $f_{G''} (k):= S_{G''}(k) + k$, and let $f^*_{G''}: =\min_k f_{G''}(k)$. Note there exists following upper bound for $f^*_{G''}$.

$$f^*_{G''} \leq S_G(M) + nM$$

This is achieved by removing the $M$-vertex separator of $G$ from $G''$ and divide $G''_{V'' \backslash S_G(M)}$ into connected components with size at most $nM$.

Now we prove that $f^*_{G''}$ has to be exactly $S_G(M) + nM$ by showing that $f_{G''} (k) > f^*_{G''}$ for $k>nM$, and for $k<nM$.

\begin{enumerate}
\item $f_{G''} (k) > f^*_{G''}$ for all $k<nM$.

For $k<nM$:
$$f_{G''} (k) \geq n^2 + k >S_G(M) +nM, $$

because the separator has to include at least one vertex from each of the $n^2$ disjoint $nM$-cliques in $G''$. This value $f_{G''} (k) $ is clearly larger than $S_G(M)+nM$ when $n\gg N>M$.

\item $f_{G''} (k) > f^*_{G''}$ for all $k>nM$.

\begin{claim}
We claim that it suffices to only consider $k$ in the form of $k = nM + n\alpha$, where $\alpha \in \mathbb{Z}_+$. i.e. for any $k>nM$, $f_{G''} (k) \geq f_{G''}(nM+n\alpha)$ for some $\alpha \in \mathbb{Z}_+$. 
\label{subclaim}
\end{claim} 
\begin{proof}[Proof of Claim \ref{subclaim}]
Call the nodes in $G$ to which each of the $n$-clique is appended to (while constructing $G''$) the \textbf{center} of the $n$-clique in $G''$. If $k$ cannot be expressed in the form of $nM + n\alpha,$ this means the corresponding separator $S$ contain some non-center nodes of the $n$-cliques in $G''$. 

If the center $\not\in S$, while some other node(s) of the clique $\in S$, there exists another $S^*$, $|S^*|<|S|$ that includes the center instead of the other node(s), and suffice to be a $k$-vertex separator. This is because after the removal of the center node, the rest of the clique can be of size at most $(n-1)$, and $k>nM>n-1$.

Suppose the center $\in S$, while some of the other node(s) of the clique also $\in S$, in order to obtain a $k$-vertex separator. Then $S^*$ that only contains center will suffice to be $k$-vertex separator, because $k>n$. 
\end{proof}

By Claim \ref{subclaim}, for any $k>nM$, $f_{G''} (k) \geq f_{G''}(nM+n\alpha)$ for some $\alpha \in \mathbb{Z}_+$. In words, there is never any incentive to include any non-center nodes of an $n$-cliques in separator $S$. Without loss of generality, $S\subseteq V_G$, and $k = nM+n\alpha \geq nM + n$. 

$$f_{G''} (nM+n\alpha) > nM+n  > S_G(M) +nM$$

when $n\gg N$.

\end{enumerate}

Summarizing 1 and 2, we conclude that
\begin{equation*}
f^*_{G''} = f_{G''}(nM) = S_G(M) + nM
\end{equation*}
\end{proof}

This gives us a polynomial algorithm to find 
any $M$-vertex separator for any graph $G$, and any value $M$.
This contradicts the fact that computing $M$-vertex separator is NP-hard. Therefore, there does not exist polynomial time algorithm for computing $\min S_G(k)+k$.
\end{proof}

